\documentclass[letterpaper, 10 pt, conference]{ieeeconf}
\IEEEoverridecommandlockouts
\usepackage{graphicx}
\usepackage{amsmath,amssymb}
\usepackage{flushend}
\usepackage{color}
\usepackage{latexsym}
\usepackage{comment}
\usepackage{url}
\usepackage{graphics}
\usepackage{pdfsync}
\usepackage{stfloats}
\usepackage{epsfig}
\usepackage{times}
\usepackage{amssymb}
\usepackage{amsmath}
\usepackage{amsfonts}
\usepackage{mathtools}
\usepackage{textcomp}
\usepackage{xcolor}
\usepackage{indentfirst}
\usepackage{epstopdf}
\usepackage{multirow}
\usepackage{float}
\usepackage{color}
\usepackage{enumerate}
\usepackage{subfigure}
\usepackage{comment}
\usepackage{cite}
\usepackage{caption}
\usepackage{bm}

\usepackage{algorithm}
\usepackage{algpseudocode}

\DeclareMathOperator*{\minimize}{minimize}

\usepackage{tcolorbox}

\title{\LARGE \bf Bilinear Data-Driven Min-Max MPC: Designing Rational Controllers via Sum-of-squares Optimization}

\author{Yifan~Xie, Julian~Berberich, Robin~Str\"{a}sser, Frank Allg\"{o}wer
	\thanks{F. Allg\"{o}wer is thankful that his work was funded by Deutsche Forschungsgemeinschaft (DFG, German Research
Foundation) under Germany’s Excellence Strategy - EXC 2075 - 390740016 and under grant 468094890.
The authors thank the International Max Planck Research School
for Intelligent Systems (IMPRS-IS) for supporting Y. Xie and thank the Graduate Academy of the SC SimTech for supporting R.~Str\"{a}sser.
	}
	\thanks{The authors are with the Institute for Systems Theory and Automatic Control, University of Stuttgart, 70550 Stuttgart, Germany.
		{\tt\small  (email: \{yifan.xie,julian.berberich,robin.straesser,
        frank.allgower\}@ist.uni-stuttgart.de}). }
}

\newtheorem{mythm}{Theorem}

\newtheorem{mylem}{Lemma}

\newtheorem{remark}{Remark}
\newtheorem{assum}{Assumption}

\begin{document}
	
	\maketitle

\begin{abstract}
We propose a data-driven min-max model predictive control (MPC) scheme to control unknown discrete-time bilinear systems. Based on a sequence of noisy input-state data, we state a set-membership representation for the unknown system dynamics. Then, we derive a sum-of-squares (SOS) program that minimizes an upper bound on the worst-case cost over all bilinear systems consistent with the data. As a crucial technical ingredient, the SOS program involves a rational controller parameterization to improve feasibility and tractability. We prove that the resulting data-driven MPC scheme ensures closed-loop stability and constraint satisfaction for the unknown bilinear system.
We demonstrate the practicality of the proposed scheme in a numerical example.
\end{abstract}

\section{Introduction}\label{sec:1}
Data-driven control leverages available data to control systems without relying on explicit mathematical models.
While data-driven control is well-established for linear systems, some methods have also been extended to certain classes of nonlinear systems \cite{martin2023guarantees,berberich2024overview}.
Bilinear systems represent a special class of nonlinear systems, where the nonlinearity arises from the product of the states and input \cite{mohler1973bilinear}.
Bilinear systems can effectively model various real-world dynamics, e.g., the dynamics of a braking car, neutron kinetics, and biochemical reactions \cite{mohler1973bilinear}.
Model-based controller design methods have been proposed for bilinear systems, e.g., based on model predictive control (MPC) \cite{bloemen2001interpolation}, linear matrix inequalities (LMIs) \cite{strasser2023control} and sum-of-squares (SOS) optimization \cite{vatani2014control}.
However, these approaches require a known mathematical model of the system.

Data-driven methods for bilinear systems have recently gained an increasing interest since general nonlinear systems can be transformed into bilinear control systems via Koopman operator theory, and then be controlled via robust or predictive control methods \cite{bold2024data, strasser2025koopman}.
Additionally, data-driven control for bilinear systems has been considered in the setting where data is affected by stochastic noise\cite{chatzikiriakos2024end}.
However, these methods rely on a two-step procedure, i.e., they first use the collected data to identify the system dynamics by solving linear regression problems, and then apply robust controller design approaches to account for possible identification errors. 
This indirect data-driven controller design needs to account for the error bounds of the system identification for rigorous design, which can be conservative.
Direct data-driven controller design methods for bilinear systems have been proposed in \cite{yuan2022data,bisoffi2020bilinear}, however,  these papers consider the setting where the system is not affected by noise.

In the recent literature, direct data-driven frameworks were proposed for linear time-invariant (LTI) systems via Willems' Fundamental Lemma \cite{willems2005note,markovsky2021review} and data informativity \cite{van2023informativity}.
This framework is initially developed for LTI systems \cite{van2020noisy,berberich2023combining,persis2020formulas}, and extended to polynomial systems \cite{martin2023inference,guo2021data,strasser2021data} and linear parameter varying systems \cite{xie2024dataLPV,verhoek2024decoupling}.
Inspired by the model-based min-max MPC \cite{kothare1996robust}, data-driven min-max MPC scheme has been proposed to consider performance objectives and constraints for LTI systems using LMIs-based methods \cite{xie2024minmaxrobust}.
Existing data-driven control approaches for polynomial systems address continuous-time systems and, in particular, they require state derivative measurements, which are hard to obtain in practice.
On the other hand, data-driven control of polynomial systems in discrete time remains unexplored in this framework.

In this paper, we consider a specific class of discrete-time polynomial systems, i.e., bilinear systems, and derive a controller design method for such systems using no explicit model knowledge but only noisy data.
We formulate an SOS program that minimizes an upper bound on the bilinear data-driven min-max MPC cost and derive a corresponding rational feedback controller.
The proposed scheme guarantees that the closed-loop system is robustly stabilized to a robust positive invariant (RPI) set around the origin and satisfies the constraints.
On a technical level, this paper extends the results from \cite{xie2024minmaxrobust}, which consider linear systems, to bilinear systems.
In contrast to existing data-driven control approaches for polynomial systems, which require derevative data, the proposed method only relies on discrete input-state measurements.
Further, a numerical example shows the effectiveness of the proposed method.

This paper is structured as follows.
Section \ref{sec:2} introduces the considered bilinear systems and the MPC problem setup. 
In Section \ref{sec:3}, we derive a set-membership characterization for bilinear systems using the available data and propose the bilinear data-driven min-max MPC scheme.
Then, we prove closed-loop robust stability and constraint satisfaction.
In Section \ref{sec:4}, we illustrate the
effectiveness of the proposed
controller design method with a numerical example.
Finally, we conclude the paper in Section \ref{sec:5}.

\textbf{Notations:} The set of natural numbers is denoted by $\mathbb{N}$.
The set of integers from $a$ to $b$ is denoted by $\mathbb{I}_{[a, b]}$.
For a symmetric matrix $P$, we write $P\succ 0$ or $P\succeq 0$ if $P$ is positive definite or positive semi-definite, respectively, where negative (semi-)definite matrices are defined analogously.
For matrices $A$ and $B$ of compatible dimensions, we abbreviate $ABA^\top$ to $AB\begin{bmatrix}\star\end{bmatrix}^\top$.
The minimal eigenvalue of the matrix $P$ is denoted by $\lambda_{\min}(P)$.
We denote by $\mathbb{R}[x, d]$ the set of all polynomials in the variable $x\in\mathbb{R}^n$ with degree $d$, i.e, $p(x)=\sum_{\alpha\in\mathbb{N}^n, |\alpha|\leq d}s_\alpha x^\alpha$ with multi-index $\alpha=\begin{bmatrix}\alpha_1 &\ldots &\alpha_n\end{bmatrix}^\top\in\mathbb{N}^n$, real coefficients $s_\alpha \in\mathbb{R}$, monomials $x^\alpha=x_1^{\alpha_1}\ldots x_n^{\alpha_n}$, $|\alpha|=\alpha_1+\ldots +\alpha_n$.
The set of all $m\times n$ matrices with elements in $\mathbb{R}[x, d]$ is denoted by $\mathbb{R}[x, d]^{m\times n}$.
A matrix $S\in\mathbb{R}[x, 2d]^{n\times n}$ is an SOS matrix in $x$ if it can be decomposed as $S=T^\top T$ for some $T\in\mathbb{R}[x, d]^{m\times n}$ and we write it as $S\in \text{SOS}[x, 2d]^n$.
A matrix $S\in\mathbb{R}[x, 2d]^{n\times n}$ is strictly SOS if $(S-\epsilon I_n)\in \text{SOS}[x, 2d]^n$ for some $\epsilon>0$ and we write it as $S\in \text{SOS}_+[x, 2d]^n$.

\section{Problem Setup}\label{sec:2}
We consider discrete-time bilinear systems
\begin{equation}\label{system}
x_{t+1}=A_s x_t+\sum_{i=1}^{m}(B_{s,i}x_t+b_{s, i}) u_{i,t}
+\omega_t,
\end{equation}
where $x_t\in\mathbb{R}^n$ is the state,  $u_{i,t}$ is the $i$th element of the input $u_t\in\mathbb{R}^m$, $\omega_t\in\mathbb{R}^n$ is the process noise.
We assume that the system matrices $A_s\in\mathbb{R}^{n\times n}$ and $B_{s, i}\in\mathbb{R}^{n\times n}$, $b_{s, i}\in\mathbb{R}^n$ for all $i\in\mathbb{I}_{[1, m]}$ are unknown.
The system dynamics \eqref{system} can be written equivalently as
\begin{equation}\label{system1}
    x_{t+1}=A_sx_t+B_su_t+\tilde{B}_s(u_t\otimes x_t)+\omega_t,
\end{equation}
where $B_s=\begin{bmatrix}b_{s, 1} &\ldots &b_{s, m}\end{bmatrix}$ and $\tilde{B}_s=\begin{bmatrix}B_{s, 1} &\ldots &B_{s, m}\end{bmatrix}$.
We assume that the noise $\omega_t$ satisfies the following assumption.

\begin{assum}\label{assumption1}
For all $t\in\mathbb{N}$, the noise $\omega_t\in\mathbb{R}^n$ satisfies $\|\omega_t\|_G\leq 1$ for a known matrix $G\succ 0$.
\end{assum}

Although the system matrices $A_s, B_s, \tilde{B}_s$ are unknown, we have a sequence of input-state measurements $(U, X)$ of length $T$ generated from the system \eqref{system1}, which we arrange in matrices 
\begin{equation}\label{data}
    U=\begin{bmatrix}u_0^d \!\!&\ldots \!\!&u_{T-1}^d\end{bmatrix},
X=\begin{bmatrix}x_0^d \!\!&\ldots \!\!&x_{T-1}^d \!\!&x_T^d\end{bmatrix}.
\end{equation}
The sequence of noise $\omega_i^d$ for all  $i\in\mathbb{I}_{[0, T-1]}$ is unknown, but satisfies Assumption \ref{assumption1}.

In this paper, our goal is to design a data-driven control law that robustly stabilizes the unknown system \eqref{system1} such that the closed-loop input and state satisfy the constraints.
To quantify the performance of the controller, we define a quadratic stage cost function $\ell(u, x)=\|u\|_R^2+\|x\|_Q^2$, where $R, Q\succ 0$.
For the controller design, we consider the input and state constraints
$\|u_t\|_{S_u}\leq 1, \|x_t\|_{S_x}\leq 1, \forall t\in\mathbb{N}$,
where $S_u\succ 0, S_x\succeq 0$.
Ellipsoidal constraints are commonly used for LMI-based min-max MPC \cite{kothare1996robust} and are well-suited to the considered setting, which includes an ellipsoidal bound on the noise.

\section{Bilinear Data-Driven Min-Max MPC using SOS Optimization}\label{sec:3}
We first characterize the set of system matrices that are consistent with the input-state data and Assumption \ref{assumption1} in \ref{sec:3.1}.
In Section \ref{sec:3.2}, we formulate an SOS program to minimize an upper bound on the bilinear data-driven min-max MPC problem and to derive a rational controller.
We prove that the resulting closed-loop system is robustly stabilized and satisfies the constraints in Section \ref{sec:3.3}.

\subsection{Data-Driven System Characterization}\label{sec:3.1}

The set of system matrices $(A, B, \tilde{B})$ consistent with the data \eqref{data} and the bound on the noise in Assumption \ref{assumption1} is defined as
$\Sigma=\{(A, B, \tilde{B}): (A, B, \tilde{B})\in\Sigma_i, \forall i\in\mathbb{I}_{[0, T-1]}\}$,
where 
\begin{equation}\nonumber
\Sigma_{i}\!=\!\!\left\{\!(A, \!B, \!\tilde{B})\!:\!
\begin{gathered}
\omega_i^d\!=\!x_{i+1}^d\!-\!Ax_i^d\!-\!Bu_i^d\!-\!\tilde{B}(u_i^d\otimes x_i^d) \\
\text{ holds for some }\omega_{i}^d
\text{ satisfying } \|\omega_i^d\|_G\!\leq\! 1
\end{gathered}
\right\}.
\end{equation}

In the following lemma, we characterize the set of $(A, B, \tilde{B})$ consistent with the data \eqref{data}.
\begin{mylem}
Suppose Assumption \ref{assumption1} hold. The set $\Sigma$ is equal to
\begin{equation}\label{Sigma}
\left\{\!\!(A, B, \tilde{B}):\!\!\!\!\!\!
\begin{gathered}
\begin{bmatrix}
I \!\!\!&A \!\!\!&B \!\!\!&\tilde{B}
\end{bmatrix}
\!\!M(\tau(x))\!\!
\begin{bmatrix}
I \!\!\!&A \!\!\!&B \!\!\!&\tilde{B}
\end{bmatrix}^\top\!\succeq\! 0,  \\
\forall \tau(x)\!=\!(\tau_0(x), \ldots, \tau_{T-1}(x)),\\ \tau_i\in \text{SOS}[x, 2\alpha],
\alpha\in\mathbb{N}, i\in\mathbb{I}_{[0, T-1]}
\end{gathered}
\right\},
\end{equation}
where
$M(\tau(x))\!=\!\!\!\sum_{i=0}^{T-1}\tau_i(x)N_i$ and 
\[N_i=\begin{bmatrix}
I \!\!\!&x_{i+1}^d\\
0 \!\!\!&-x_i^d\\
0 \!\!\!&-u_i^d\\
0 \!\!\!&-(u_i^d \!\otimes\! x_i^d)
\end{bmatrix}\!\!\!
\begin{bmatrix}
    G^{-1} \!\!\!\!\!&0\\
    0 \!\!\!\!\!&-I
\end{bmatrix}\!\!\!
\begin{bmatrix}
I \!\!\!&x_{i+1}^d\\
0 \!\!\!&-x_i^d\\
0 \!\!\!&-u_i^d\\
0 \!\!\!&-(u_i^d \!\otimes\! x_i^d)
\end{bmatrix}^\top.\]
\end{mylem}
\begin{proof}
    To characterize the set $\Sigma_i$ using data, we replace $\omega_i^d$ by $x_{i+1}^d-Ax_i^d-Bu_i^d-\tilde{B}(u_i^d\otimes x_i^d)$ in $\|\omega_i^d\|_G\leq 1$.
It is then straightforward to show that the set $\Sigma_i$ is equal to 
\begin{equation}\label{sigma_i}
\left\{\begin{gathered}
(A, B, \tilde{B}):
\begin{bmatrix}
I \!\!\!&A \!\!\!&B \!\!\!&\tilde{B}
\end{bmatrix}
N_i
\begin{bmatrix}
I \!\!\!&A \!\!\!&B \!\!\!&\tilde{B}
\end{bmatrix}^\top\!\succeq\! 0
\end{gathered}
\right\},
\end{equation}
Similar to \cite{berberich2023combining} and using the characterization of the set $\Sigma_i$ in \eqref{sigma_i}, the set $\Sigma$ is equal to \eqref{Sigma}.
\end{proof}

\subsection{Bilinear Data-driven Min-Max MPC Scheme}\label{sec:3.2}

At time $t$, given input-state data \eqref{data} and an initial state $x_t\in\mathbb{R}^n$, the bilinear data-driven min-max MPC optimization problem is formulated as
\begin{subequations}\label{mpc}
\begin{align}
J_\infty^*&(x_t)=\min_{\bar{u}(t)}\max_{(A, B, \tilde{B})\in\Sigma}\sum_{k=0}^{\infty}\ell(\bar{u}_k(t), \bar{x}_k(t))\label{mpc:obj}\\
\text{s.t.}\quad &\bar{x}_{k+1}(t)\!=\!A\bar{x}_k(t)\!+\!B\bar{u}_k(t)\!+\!\tilde{B}(\bar{u}_k(t)\!\otimes\! \bar{x}_k(t)), \label{mpc:con1}\\
&\bar{x}_0(t)=x_t, \label{mpc:con2}\\
&\|\bar{x}_k(t)\|_{S_x}\leq 1, \forall (A, B, \tilde{B})\in\Sigma, k\in\mathbb{N},\label{mpc:con3}\\
&\|\bar{u}_k(t)\|_{S_u}\leq 1, \forall k\in\mathbb{N}.\label{mpc:con4}
\end{align}
\end{subequations}
In problem \eqref{mpc}, the optimization variable $\bar{u}(t)$ is a sequence of predicted inputs. 
The predicted state is initialized as $\bar{x}_0(t)=x_t$ in \eqref{mpc:con2} and the nominal system dynamics \eqref{mpc:con1} without the process noise $\bar{\omega}_k(t)$ is used for state prediction.
The objective of problem \eqref{mpc} is to minimize the worst-case infinite-horizon stage costs over all system matrices in $\Sigma$.
Moreover, input and state constraints are imposed in \eqref{mpc:con3}-\eqref{mpc:con4} for any $(A, B, \tilde{B})\in\Sigma$.

To effectively solve the problem \eqref{mpc} and obtain a tractable solution, we restrict our attention to a rational state-feedback control law.
According to \cite{vatani2014control}, an unstable discrete-time bilinear system can be globally stabilized by a rational control law if and only if the degrees of the polynomial in the numerator and the polynomial in the denominator are the same.
Therefore, we specifically focus on designing a rational control law \cite{strasser2024koopman}
\begin{equation}
    u(x)=\tfrac{1}{d(x)}K(x)x,\label{statefeedback}
\end{equation}
where $d\in \text{SOS}_{+}[x, 2\alpha]$ and $K\in\mathbb{R}[x, 2\alpha-1]^{m\times n}$.

In the following, we formulate an SOS program to find a rational state-feedback control law \eqref{statefeedback} that minimizes an upper bound on the optimal cost of \eqref{mpc}.
Given the state $x_t\in\mathbb{R}^n$, the input-state data \eqref{data}, and a constant $c>\lambda_{\min}(Q)$, we formulate the SOS program \eqref{sdp}, where $\Phi(x)=\begin{bmatrix}M_RL(x)\\ M_QHd(x)\end{bmatrix}$ with $M_R^\top M_R=R$, $M_Q^\top M_Q =Q$.
The optimal solution of \eqref{sdp} given the current state $x_t$ is denoted by $\gamma_{x_t}^\star, H_{x_t}^\star, L_{x_t}^\star, {d}_{x_t}^\star, \tau_{x_t}^\star$.
The corresponding optimal rational state-feedback controller is defined as $u^\star(x_t)=\frac{1}{d^\star_{x_t}(x_t)}K_{x_t}^\star(x_t)x_t$ with $K_{x_t}^\star(x_t)=L_{x_t}^\star(x_t)(H_{x_t}^\star)^{-1}$ and we consider the Lyapunov function $V(x_t)=\|x_t\|_{P_{x_t}^\star}^2$ with $P_{x_t}^\star=\gamma_{x_t}^\star (H_{x_t}^\star)^{-1}$.

\begin{figure*}
\begin{subequations}\label{sdp}
\begin{align}
    &\minimize_{\gamma>0, H\in\mathbb{R}^{n\times n}, L\in \mathbb{R}[x, 2\alpha-1]^{m\times n}, d\in \text{SOS}_+[x, 2\alpha], \tau\in\mathbb{R}[x, 2\alpha]^T}\gamma\label{sdp:obj}\\
    \text{s.t. }&\begin{bmatrix}1 &x_t^\top \\ x_t &H\end{bmatrix}\succeq 0, \tau(x)=(\tau_0(x), \ldots, \tau_{T-1}(x)), \tau_i\in\text{SOS}[x, 2\alpha], \forall i\in\mathbb{I}_{[0, T-1]},\label{sdp:con1}\\
    &\begin{bmatrix}
        \begin{bmatrix}
            d(x)(H-\frac{\gamma}{c}I) &0\\
            0 &0
        \end{bmatrix}-M(\tau(x)) &\begin{bmatrix} 0\\ Hd(x)\\
            L(x)\\
            L(x)\otimes x
            \end{bmatrix} &0 \\
           \begin{bmatrix}0 &Hd(x) &L(x)^\top&
            (L(x)\otimes x)^\top 
            \end{bmatrix} &d(x)H &\Phi(x)^\top\\
            0 &\Phi(x) &\gamma d(x) I
    \end{bmatrix}\in \text{SOS}[x, 2\alpha]^{4n+2m+mn},\label{sdp:con2}\\
    &\begin{bmatrix}
        H &H\\
        H &S_x^{-1}
    \end{bmatrix}\succeq 0, 
    \begin{bmatrix}
        d(x) H &L(x)^\top\\
        L(x) &d(x) S_u^{-1}
    \end{bmatrix}\in \text{SOS}[x, 2\alpha]^{n+m},\label{sdp:con3}
\end{align}
\end{subequations}
\rule{\textwidth}{0.4pt} 
\vspace{-30pt}
\end{figure*}

In the following theorem, we prove that the problem \eqref{sdp:obj}-\eqref{sdp:con2} minimizes an upper bound on the optimal cost of the problem \eqref{mpc:obj}-\eqref{mpc:con2}.

\begin{mythm}
Suppose that Assumption \ref{assumption1} holds and let a state $x_t\in\mathbb{R}^n$ and a constant $c>\lambda_{\min}(Q)$ be given. If there exist $\alpha\in\mathbb{N}$, $H\in\mathbb{R}^{n\times n}$, $L\in\mathbb{R}[x, 2\alpha-1]^{m\times n}$, $d\in \text{SOS}_+[x, 2\alpha], \tau\in\mathbb{R}[x, 2\alpha]^T, \gamma>0$ such that the constraints in \eqref{sdp:con1}-\eqref{sdp:con2} hold, then the optimal cost of \eqref{mpc:obj}-\eqref{mpc:con2} is guaranteed to be upper bounded by $\gamma$.
\end{mythm}
\begin{proof}
We apply the Schur complement to $\gamma d(x)$ in the matrix $\begin{bmatrix}
d(x)H &\Phi(x)^\top\\
\Phi(x) &\gamma d(x) I\end{bmatrix}$ in \eqref{sdp:con2} and obtain \eqref{thm11}.
Then applying the Schur complement to the right lower block in \eqref{thm11}, we obtain that \eqref{thm1:1} and 
$d(x)H-\tfrac{1}{\gamma d(x)}\Phi(x)^\top \Phi(x)\succeq 0$
for all $x\in\mathbb{R}^n$.
\begin{figure*}
\begin{equation}\label{thm11}
    \begin{bmatrix}
        \begin{bmatrix}
            d(x)(H-\frac{\gamma}{c}I) &0\\
            0 &0
        \end{bmatrix}-M(\tau(x)) &\begin{bmatrix}0\\ Hd(x)\\
            L(x)\\
            L(x)\otimes x
            \end{bmatrix}\\
           \begin{bmatrix}0 &d(x)H &L(x)^\top&
            (L(x)\otimes x)^\top 
            \end{bmatrix} &d(x)H-\frac{1}{\gamma d(x)}\Phi(x)^\top \Phi(x)
    \end{bmatrix}\in \text{SOS}[x, 2\alpha]^{3n+m+mn},
\end{equation}
\begin{equation}\label{thm1:1}
\begin{matrix}\underbrace{
    \begin{bmatrix}
        d(x)(H-\frac{\gamma}{c}I) &0\\
        0 &-\begin{bmatrix}
            Hd(x)\\
            L(x)\\
            L(x)\otimes x
        \end{bmatrix}
        (d(x) H\!-\!\frac{1}{\gamma d(x)}\Phi(x)^\top \Phi(x))^{-1}
        \begin{bmatrix}
            Hd(x)\\
            L(x)\\
            L(x)\otimes x
        \end{bmatrix}^\top
    \end{bmatrix}}\\S\end{matrix}-M(\tau(x))\succeq 0
\end{equation}
\vspace{-1pt}
\addtocounter{equation}{1}
\begin{equation}\label{thm1:3}
    d(x) (H-\tfrac{\gamma}{c}I)-[AHd(x)\!+\!BL(x)\!+\!\tilde{B}(L(x)\otimes x)](d(x)H\!-\tfrac{1}{\gamma d(x)}\Phi(x)^\top \Phi(x))^{-1}[AHd(x)\!+\!BL(x)\!+\!\tilde{B}(L(x)\otimes x)]^\top\succeq 0
\end{equation}
\begin{equation}\label{thm1:4}
[AHd(x)+BL(x)+\tilde{B}(L(x)\otimes x)]^\top (d(x)(H-\tfrac{\gamma}{c}I))^{-1}[\star]-d(x)H+\tfrac{1}{\gamma d(x)}\Phi(x)^\top \Phi(x)\preceq 0
\end{equation}
\vspace{-10pt}
\addtocounter{equation}{+1}
\begin{equation}
\label{thm1:41}
[A\!+\!B\tfrac{1}{d(x)}L(x)H^{-1}\!\!+\tilde{B}\tfrac{1}{d(x)}(L(x)\otimes x)H^{-1}]^\top \!\gamma(H\!-\tfrac{\gamma}{c}I)^{-1}\![\star]
-\gamma H^{-1}+\tfrac{1}{d(x)^2}H^{-1}\Phi(x)^\top \Phi(x) H^{-1}\preceq 0
\end{equation}
\rule{\textwidth}{0.4pt} 
\vspace{-25pt}
\end{figure*}
Recall the QMI in \eqref{Sigma}. Pre- and post-multiplying \eqref{thm1:1} with $\begin{bmatrix}I &A &B &\tilde{B}\end{bmatrix}$ and its transpose, respectively, and using the S-procedure on the resulting inequality imply 
\addtocounter{equation}{-5}
\begin{equation}\label{thm1:2}
    \begin{bmatrix}
        I &A &B &\tilde{B}
    \end{bmatrix}
    S
    \begin{bmatrix}
        I &A &B &\tilde{B}
    \end{bmatrix}^\top\succeq 0
\end{equation}
for all $(A, B, \tilde{B})\in \Sigma$.
Expanding \eqref{thm1:2} yields equivalently \eqref{thm1:3}.
Applying the Schur complement twice to \eqref{thm1:3}, we obtain \eqref{thm1:4} and
\addtocounter{equation}{2}
\begin{equation}\label{thm1:31}
    H-\tfrac{\gamma}{c}I\succeq 0.
\end{equation}
Multiplying \eqref{thm1:4} from left and right by $\sqrt{\frac{\gamma}{d(x)}}H^{-1}$, we obtain \eqref{thm1:41}.
Let $P=\gamma H^{-1}$, $K(x)=L(x)H^{-1}$ and $\Pi(x)=A+B\frac{1}{d(x)}K(x)+\tilde{B}(\frac{1}{d(x)}K(x)\otimes x)$.
Since $(L(x)\otimes x)H^{-1}=(L(x)H^{-1})\otimes x=K(x)\otimes x$, \eqref{thm1:41} implies 
\addtocounter{equation}{+1}
\begin{equation}\label{thm1:5}
        \Pi(x)^\top\! (P^{-1}\!\!-\!\tfrac{1}{c}I)^{-1}\Pi(x) \!-\! P\! 
        +Q\!+\!\tfrac{1}{d(x)^2}K(x)^
    \top RK(x)\preceq 0,
\end{equation}
where we use the definition of $\Phi(x)$ to deduce it.
Using the Woodbury matrix identity \cite{MR0038136}, we have $(P^{-1}-\frac{1}{c}I)^{-1}=P-P(P-cI)^{-1}P$.
Thus, \eqref{thm1:5} is equivalent to
\begin{multline}\label{thm1:6}
        \Pi(x)^\top [P-P(P-cI)^{-1}P]\Pi(x) - P\\+Q+\tfrac{1}{d(x)^2}K(x)^
    \top RK(x)\preceq 0.
\end{multline}
Using the same arguments as in \cite[equations (15)-(17)]{xie2024minmaxrobust}, \eqref{thm1:31} is equivalent to $P-cI\preceq 0$.
Using the Schur complement, \eqref{thm1:6}  and $P-cI\preceq 0$ imply
\begin{equation}\label{thm1:7}
\begin{bmatrix}
\Pi(x)^\top \!P\Pi(x)\!\!-\!\!P\!+\!Q\!+\!\frac{1}{d(x)^2}K(x)^\top\! R K(x) \!\!\!&\Pi(x)^\top P\\
P \Pi(x) \!\!\!&P-cI
\end{bmatrix}\preceq 0.
\end{equation}
This implies for all $(A, B, \tilde{B})\in\Sigma$, we have 
\begin{equation}\label{thm1:8}
    \Pi(x)^\top P\Pi(x) -P+Q+\tfrac{1}{d(x)^2}K(x)^\top R K(x)\preceq 0.
\end{equation}
Multiplying \eqref{thm1:8} from the left and the right by $x^\top$ and $x$, respectively, we obtain
\begin{equation}\label{thm1:9}
    \|x_+\|_P^2-\|x\|_P^2\leq -\ell(u, x).
\end{equation}
for all $x\in\mathbb{R}^n$ and all $(A, B, \tilde{B})\in\Sigma$ with $x_+=\Pi(x) x$ and $u=\frac{1}{d(x)}K(x)x$, where we use $(K(x)\otimes x)x=K(x)x\otimes x$.
Using \eqref{thm1:9} and the same reasoning as in \cite[equation (21)-(24)]{xie2024minmaxrobust}, which also applies to bilinear systems, we conclude that $\gamma$ is an upper bound on the optimal cost of \eqref{mpc}.
\end{proof}

The SOS program \eqref{sdp} is linear in the decision variables $\gamma$, $H$, $L$ and $\tau$ when the variable $d$ is fixed.
The parameter $d$ can be pre-selected as a tuning parameter, which allows \eqref{sdp} to be solved using convex optimization techniques.
An iterative procedure can be used to iteratively update $d$ and other decision variables.
The choice of $c$ influences the closed-loop performance.
Both excessively small or too large values of $c$ can lead to poor closed-loop performance. 
Besides, a smaller value of $c$ results in a smaller RPI set, but also leads to a smaller feasible region.
A detailed discussion on selecting an appropriate $c$ to achieve better closed-loop behavior can be found in \cite[Remark 3]{xie2024minmaxrobust}.

The SOS program \eqref{sdp} is solved in a receding-horizon manner, see Algorithm~1.
To ensure convexity of \eqref{sdp}, we fix $d$ throughout the algorithm, but we note that $d$ can be adapted online to improve performance.
In Algorithm~1, we define $\tilde{K}=K_{x_{t-1}}^\star, \tilde{P} = P_{x_{t-1}}^\star, \tilde{d}=d$, where $t-1$ is the last time step when $\gamma_{x_{t-1}}^\star>\tfrac{c^2}{\lambda_{\min}(Q)\lambda_{\min}(G)}$.
We will later show in Theorem \ref{theorem2} that the closed-loop system converges to the RPI set $\mathcal{E}_{\mathrm{RPI}}=\{x\in\mathbb{R}^n:\|x\|_{\tilde{P}}^2\leq \tfrac{c^2}{\lambda_{\min}(Q)\lambda_{\min}(G)}\}$ for any initial state inside the region of attraction (ROA) $\mathcal{E}_{\mathrm{ROA}}=\{x\in\mathbb{R}^n:\|x\|_{P_{x_0}^\star}\leq \gamma_{x_0}^\star\}$.
When the state is outside the RPI set, i.e., $\gamma_{x_t}^\star>\tfrac{c^2}{\lambda_{\min}(Q)\lambda_{\min}(G)}$, we implement the first computed input, measure the state at the next time step, and resolve the problem \eqref{sdp}.
When the state is inside the RPI set, we stop solving the problem \eqref{sdp} and directly apply the input  $u_t=\frac{1}{\tilde{d}(x_t)}\tilde{K}(x_t)x_t$ from this time onward.

\begin{remark}
When removing the first constraint in \eqref{sdp:con1} as well as the constraints \eqref{sdp:con3}, the SOS program \eqref{sdp} yields a robust linear quadratic regulator for bilinear systems.
In contrast to existing results for continuous-time systems \cite{guo2021data, strasser2021data}, the proposed approach only requires input-state measurements but no derivative data.
\end{remark}

\begin{algorithm}[htb]
\begin{algorithmic}[1]
\caption{\!Data-driven min-max MPC scheme.\!\!\!}
    \State \algorithmicrequire{ $U, X$, $Q, R, S_x, S_u$, $c$, $G$, $d$}\;
    \State At time $t=0$, measure state $x_0$\;
    \State Solve the problem \eqref{sdp} for a given $d^\star_{x_t}=d$\;
    \State Apply the input $u_t=\frac{1}{d^\star_{x_t}(x_t)}K_{x_t}^\star(x_t)x_t$\;
    \State Set $t=t+1$,  measure state $x_t$\;
    \State Solve the SOS problem \eqref{sdp} for a given $d^\star_{x_t}=d$\;
    \If{$\gamma_{x_t}^\star>\frac{c^2}{\lambda_{\min}(Q)\lambda_{\min}(G)}$}\;
    \State Apply the input $u_t=\frac{1}{d^\star_{x_t}(x_t)}K_{x_t}^\star(x_t)x_t$\;
    \State Set $t=t+1$,  measure state $x_t$, go back to 6\;
    \ElsIf{$\gamma_{x_t}^\star\leq\frac{c^2}{\lambda_{\min}(Q)\lambda_{\min}(G)}$}
    \State Set $\tilde{K}=K_{x_{t-1}}^\star, \tilde{P} = P_{x_{t-1}}^\star, \tilde{d}=d$\;
    \State Apply the input $u_t=\frac{1}{\tilde{d}(x_t)}\tilde{K}(x_t)x_t$\;
    \State Set $t=t+1$,  measure state $x_t$, go back to 12\;
    \EndIf
    \label{algorithm:robust}
\end{algorithmic}
\end{algorithm}

\subsection{Closed-loop Guarantees}\label{sec:3.3}
In the following theorem, we prove that problem \eqref{sdp} is recursively feasible and 
the resulting closed-loop system is robustly stabilized to an RPI set for any $(A, B, \tilde{B})\in\Sigma$.
Further, we show that the input and state constraints are satisfied for the closed-loop trajectory.

\begin{mythm}\label{theorem2}
Suppose Assumption \ref{assumption1} holds.
If the SOS program \eqref{sdp} is feasible at time $t=0$ and $\gamma_{x_0}^\star\geq \frac{c^2}{\lambda_{\min}(Q)\lambda_{\min}(G)}$, then
\begin{enumerate}[i)]
\item the optimization problem \eqref{sdp} is feasible for any states $x_t\in\mathcal{E}_{\mathrm{ROA}}\backslash\mathcal{E}_{\mathrm{RPI}}$;
\item the set $\mathcal{E}_{\mathrm{RPI}}$ is robustly stabilized for the closed-loop system resulting from Algorithm~1 with any $(A, B, \tilde{B})\in\Sigma$;
\item the closed-loop trajectory resulting from Algorithm~1 with any $(A, B, \tilde{B})\in\Sigma$ satisfies the constraints, i.e., $\|u_t\|_{S_u}\leq 1, \|x_t\|_{S_x}\leq 1$ for all $t\in \mathbb{N}$.
\end{enumerate}
\end{mythm}
\begin{proof}
Parts of the proof are analogous to the proof 
in \cite[Theorem 2]{xie2024minmaxrobust}, so we only state the parts different from \cite{xie2024minmaxrobust}.
The proof is composed of two parts. 
Part I proves the recursive feasibility of the problem \eqref{sdp} and robust stability of the closed-loop system.
Part II proves that the input and state constraints are satisfied for the closed-loop system.

\textbf{Part I: }Assuming problem \eqref{sdp} is feasible at time $t$, the inequality \eqref{thm1:7} holds for ${d}_{x_t}^\star, P_{x_t}^\star, K_{x_t}^\star$ and for any $(A, B, \tilde{B})\in\Sigma$ for all $x\in\mathbb{R}^n$.
Pre- and post-multiplying \eqref{thm1:7} by $\begin{bmatrix}x_t^\top &\omega_t^\top\end{bmatrix}$ and its transpose, respectively, the inequality 
\begin{multline}\label{thm2:1}
    [\Pi^\star(x_t)x_t+\omega_t]^\top P_{x_t}^\star [\Pi^\star(x_t)x_t+\omega_t]-
    x_t^\top P_{x_t}^\star x_t\leq \\
    -x_t^\top (Q+\tfrac{1}{d^\star_{x_t}(x_t)^2}K^\star_{x_t}(x_t)^\top R K^\star_{x_t}(x_t))x_t+c\omega_t^\top \omega_t
\end{multline}
with $\Pi^\star(x)=A+B\frac{1}{d_x^\star(x)}K_x^\star(x)+\tilde{B}\frac{1}{d_x^\star(x)}(K_x^\star(x)\otimes x)$ holds for all $x_t\in\mathbb{R}^n, \omega_t\in\mathbb{R}^n$, and $(A, B, \tilde{B})\in\Sigma$.
Further, by exploiting that $\omega_t$ satisfies Assumption \ref{assumption1}, we have $\omega_t^\top\omega_t\leq \frac{\|\omega_t\|_G^2}{\lambda_{\min}(G)}\leq \frac{1}{\lambda_{\min}(G)}$.
Since $(K^\star_{x_t}(x_t)\otimes x_t)x_t=(K^\star_{x_t}(x_t)x_t)\otimes x_t$ and $x_{t+1}=\Pi^\star(x_t)x_t+\omega_t$, \eqref{thm2:1} yields
\begin{equation}
\begin{aligned}
&\|x_{t+1}\|_{P_{x_t}^\star}^2-\|x_t\|_{P_{x_t}^\star}^2
\\\leq &-x_t^\top (Q\!+\!\tfrac{1}{d^\star_{x_t}(x_t)^2}K^\star_{x_t}(x_t)^\top R K^\star_{x_t}(x_t))x_t\!+\!\tfrac{c}{\lambda_{\min}(G)}\label{thm2:2}
\\\leq &-\|x_t\|_Q^2+\tfrac{c}{\lambda_{\min}(G)}.
\end{aligned}
\end{equation}
From here on, recursive feasibility and stability can be shown analogously to \cite[Theorem 2]{xie2024minmaxrobust}.

\textbf{Part II: }
We first assume that the state $x_t\in\mathcal{E}_{\mathrm{ROA}}\backslash\mathcal{E}_{\mathrm{RPI}}$.
Using the Schur complement, 
the second SOS condition in \eqref{sdp:con3} yields
\begin{equation}\label{input:nominal_lmi4}
d(x)H-\tfrac{1}{d(x)}L(x)^{\top} S_u L(x)\succeq 0.
\end{equation}
Multiplying both sides of \eqref{input:nominal_lmi4} with $\sqrt{\frac{1}{d(x)}}H^{-1}$, we obtain 
\begin{equation}
    \tfrac{1}{\gamma}P-\tfrac{1}{d(x)^2}K(x)^\top S_u K(x)\succeq 0.\label{thm2:3}
\end{equation}
Using \eqref{thm2:3} and $1-\gamma^{-1}\gamma=0$ yields
\begin{equation}\label{thm2:22}
    \begin{bmatrix}
        -\tfrac{1}{d(x)^2}K(x)^\top S_u K(x) &0\\
        0 &1
    \end{bmatrix}-\gamma^{-1}
    \begin{bmatrix}
        -P &0\\
        0 &\gamma
    \end{bmatrix}\succeq 0.
\end{equation}
Pre- and post-multiplying \eqref{thm2:22} with $\begin{bmatrix}x^\top &1\end{bmatrix}$ and its transpose, respectively, we obtain
\begin{equation}\label{thm2:23}
    \begin{bmatrix}
        x\\1
    \end{bmatrix}^\top\!\!\!\!
    \left(\!\begin{bmatrix}
        -\tfrac{1}{d(x)^2}K(x)^\top S_u K(x) &0\\
        0 &1
    \end{bmatrix}\!\!-\gamma^{-1}\!\!
    \begin{bmatrix}
        -P \!\!&0\\
        0 \!\!&\gamma
    \end{bmatrix}\!\right)\!\!\begin{bmatrix}
        x\\1
    \end{bmatrix}\!\succeq\! 0
\end{equation}
for all $x\in\mathbb{R}^n$.
Since $d^\star_{x_t}$, $K^\star_{x_t}$, $\gamma^\star_{x_t}$, $P^\star_{x_t}$ is the optimal solution of \eqref{sdp}, the inequality \eqref{thm2:23} holds for $d=d^\star_{x_t}$, $K=K^\star_{x_t}$, $\gamma=\gamma^\star_{x_t}$, $P=P^\star_{x_t}$.
This implies for any state $x$ satisfying
\begin{equation}\nonumber
    \begin{bmatrix}
        x\\
        1
    \end{bmatrix}^\top
    \begin{bmatrix}
        -P^\star_{x_t} &0\\
        0 &\gamma^\star_{x_t}
    \end{bmatrix}
    \begin{bmatrix}
        x\\
        1
    \end{bmatrix}
    \geq 0
\end{equation}
the inequality 
\begin{equation}\nonumber
    \begin{bmatrix}
        x\\
        1
    \end{bmatrix}^\top
    \begin{bmatrix}
        -\frac{1}{d^\star_{x_t}(x)^2}K^\star_{x_t}(x)^\top S_u K^\star_{x_t}(x) &0\\
        0 &1
    \end{bmatrix}
    \begin{bmatrix}
        x\\
        1
    \end{bmatrix}
    \geq 0.
\end{equation}
Thus, we have
$\max\limits_{x\in\mathcal{E}_t}\|\frac{1}{d^\star_{x_t}(x)}K^\star_{x_t}(x) x\|_{S_u}^2\leq 1$,
where $\mathcal{E}_t=\{x\in\mathbb{R}^n: \|x\|_{P_{x_t}^\star}\leq \gamma^\star_{x_t}\}$.
According to the arguments in \cite[Theorem 2]{xie2024minmaxrobust}, for any state $x_t\in\mathcal{E}_t$ and for any $\omega_t$ satisfying Assumption \ref{assumption1}, the state $x_{t+1}=\Pi^\star(x_t)x_t+\omega_t$ lies inside the set $\mathcal{E}_t$ at the next time step for any $(A, B, \tilde{B})\in\Sigma$.
Thus, the input constraint is satisfied for $u_t=\frac{1}{d^\star_{x_t}(x_t)}K^\star_{x_t}(x_t)x_t$.
The remaining parts of the proof regarding input and state constraint satisfaction for states within the RPI set follow the same reasoning as in \cite[Theorem 2]{xie2024minmaxrobust}.
\end{proof}

\begin{remark}
Theorem \ref{theorem2} shows that the closed-loop system converges to the RPI set $\mathcal{E}_{\mathrm{RPI}}$ for any $(A, B, \tilde{B})\in\Sigma$ if the SOS program \eqref{sdp} is feasible at time $t=0$ and $\gamma_{x_0}^\star\geq \frac{c^2}{\lambda_{\min}(Q)\lambda_{\min}(G)}$.
Besides, the closed-loop system resulting from Algorithm~1 satisfies the input and state constraints.
Compared with \cite{xie2024minmaxrobust}, this work extends the theoretical results of data-driven min-max MPC scheme from LTI systems to bilinear systems.
\end{remark}

\section{Simulation}\label{sec:4}
In this section, we illustrate the proposed data-driven rational controller design scheme with a numerical example.
We consider a zone temperature process used for building control \cite{huang2011model} with the system dynamics
\begin{equation}
    x_{t+1}=x_t+T_sV_z^{-1}u_t(T_0-x_t)+\omega_t,
\end{equation}
where $x_t$ is the zone temperature and $u_t$ is the air volume flow rate.
In the simulation, we choose the parameters as $V_z=2, T_0=-1, T_s=1$ such that the system matrices are given as $A_s=1, B_s=-0.5, \tilde{B}_s=-0.5$.
We assume that the system matrices are unknown, but a sequence of input-state trajectory of length $T=100$ is available, and the process noise $\omega_t$ satisfies Assumption \ref{assumption1} with $G=1\times 10^6$.

In the bilinear data-driven min-max MPC problem, the weighting matrices of the stage cost are set to $Q=R=0.01$.
The initial state is chosen as $x_0=-0.1$.
The user-defined parameter $c$ is chosen as $5$, the denominator $d$ is chosen as $0.01+(1+x)^{2\alpha}$ with $\alpha=2$.
To analyze the effect of different constraint matrices, we implement the proposed bilinear data-driven min-max MPC scheme for different values $S_u\in\{500, 1000, 1500, 5000\}$ and keep $S_x=100$ fixed.
Fig.~\ref{trajectory_constraint} presents the input and state trajectories of the closed-loop system resulting from the proposed bilinear data-driven min-max MPC scheme.
The input and state constraints are satisfied for all four closed-loop trajectories.
However, with larger values of $S_u$, the resulting input becomes more restricted and the closed-loop state trajectory converges slower to the RPI set.
We emphasize that the input always remains above the constraint bound.
This is because the constraint reformulation in \eqref{sdp:con3}, which use the S-procedure, makes the reformulated constraints more conservative.

We now choose the denominator $d$ as before but with varying values of $\alpha$.
The matrices for the input and state constraints are chosen as $S_u=500$ and $S_x=100$.
Other parameters such as the weighting matrix are kept unchanged.
Fig.~\ref{cost} illustrates the closed-loop cost and computation time of the proposed scheme for different values of $\alpha$.
As $\alpha$ increases, the closed-loop cost decreases, but the computation time increases.
When $\alpha=1$, the SOS program \eqref{sdp} is infeasible.
This is because a larger $\alpha$ leads to a larger feasible region for controller design but also increases the computational complexity.

\begin{figure}
    \centering
    \includegraphics[width=0.5\textwidth]{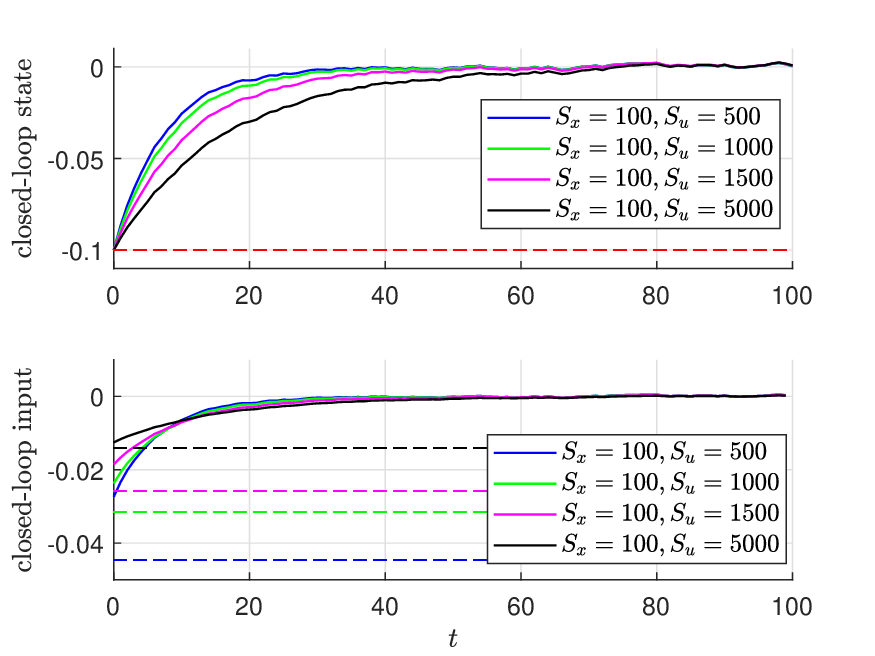}
    \caption{Closed-loop state and input trajectories of the proposed bilinear data-driven min-max MPC scheme with different input constraints.}
    \label{trajectory_constraint}
\end{figure}

\begin{figure}
    \centering
    \includegraphics[width=0.5\textwidth]{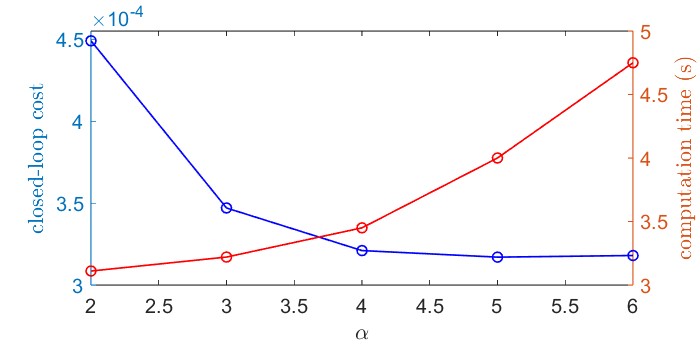}
    \caption{Closed-loop cost and computation time with different values of $\alpha$.}
    \label{cost}
    \vspace{-30pt}
\end{figure}

\section{Conclusion}\label{sec:5}

In this paper, we proposed a bilinear data-driven min-max MPC scheme for unknown discrete-time bilinear systems using noisy input-state data.
We reformulated the bilinear data-driven min-max MPC problem as an SOS program and proposed a receding-horizon algorithm that repeatedly solves the SOS program and derives a rational controller.
The resulting closed-loop system is proven to be robustly stabilized to an RPI set and satisfy input and state constraints.
Simulation results demonstrate the effectiveness of the proposed scheme.
In the future, we plan to extend this scheme to other classes of nonlinear systems such as, e.g., more general polynomial systems.

\bibliographystyle{unsrt}
\bibliography{main}

\end{document}